\documentclass[a4paper]{article}

\usepackage{a4wide,amsmath,amssymb,amsthm,hyperref,multicol,tikz-cd}

\newcommand{\converse}{\smallsmile}
\newcommand{\convex}[1]{{#1}{\updownarrow}}
\newcommand{\dom}{\mathit{dom}}
\newcommand{\down}[1]{{#1}{\downarrow}}
\newcommand{\dual}[1]{{#1}^\mathsf{d}}
\newcommand{\EM}{\mathrm{EM}}
\newcommand{\eqem}{\mathrel{=_\updownarrow}}
\newcommand{\eqh}{\mathrel{=_\downarrow}}
\newcommand{\eqs}{\mathrel{=_\uparrow}}
\newcommand{\Ho}{\mathrm{H}}
\newcommand{\icpl}[1]{{\sim}{#1}}
\newcommand{\Id}{\mathit{Id}}
\newcommand{\ii}{\Cap}
\newcommand{\iI}{\raisebox{-0.5ex}{\Large$\ii$}}
\newcommand{\iiatoms}{{\mathsf{A}_\ii}}
\newcommand{\iione}{{1_\ii}}
\newcommand{\iu}{\Cup}
\newcommand{\iU}{\raisebox{-0.5ex}{\Large$\iu$}}
\newcommand{\iuatoms}{{\mathsf{A}_\iu}}
\newcommand{\iuone}{{1_\iu}}
\newcommand{\kleisli}{\Pow}
\newcommand{\Mult}{M}
\newcommand{\Pow}{\mathcal{P}}
\newcommand{\Rel}{\mathbf{Rel}}
\newcommand{\rto}{\leftrightarrow}
\newcommand{\seq}{\ast}
\newcommand{\seqint}{\odot}
\newcommand{\Set}{\mathbf{Set}}
\newcommand{\Sm}{\mathrm{S}}
\newcommand{\subem}{\mathrel{\sqsubseteq_\updownarrow}}
\newcommand{\subh}{\mathrel{\sqsubseteq_\downarrow}}
\newcommand{\subs}{\mathrel{\sqsubseteq_\uparrow}}
\newcommand{\syq}[2]{{#1} \div {#2}}
\newcommand{\up}[1]{{#1}{\uparrow}}

\newtheorem{theorem}{Theorem}[section]
\newtheorem{proposition}[theorem]{Proposition}
\newtheorem{lemma}[theorem]{Lemma}

\theoremstyle{definition}
\newtheorem{example}[theorem]{Example}
\newtheorem{remark}[theorem]{Remark}

\begin{document}

\title{On the Inner Structure of Multirelations}
\author{Hitoshi Furusawa, Walter Guttmann and Georg Struth}
\maketitle

\begin{abstract}
  Binary multirelations form a model of alternating nondeterminism useful for analysing games, interactions of computing systems with their environments or abstract interpretations of probabilistic programs.
  We investigate this alternating structure with inner or demonic and outer or angelic choices in a relation-algebraic language extended with specific operations on multirelations that relate to the inner layer of alternation.
\end{abstract}

\section{Introduction}
\label{section.introduction}

This is the first article in a trilogy on the inner structure of multirelations, the determinisation of such relations~\cite{FurusawaGuttmannStruth2023b} and their algebras of modal operators~\cite{FurusawaGuttmannStruth2023c}.

Multirelations -- morphisms of type $X \rto \Pow Y$ in the category $\Rel$ -- are models of alternating nondeterminism.
Elements $(a,B)$, $(a,C)$ of a multirelation can be interpreted as an outer nondeterministic or angelic choice between the subsets $B$ or $C$ of $Y$ that depends on the element $a$ of $X$, or as an outer nondeterministic evolution of a system from state $a$ into the sets of states $B$ or $C$.
An element $(a,B)$, in turn, can model the inner nondeterministic or demonic choices between the elements of $B$ that depend on $a$, or an inner nondeterministic evolution from state $a$ to any state in $B$.
Multirelations have therefore been used as semantics for logics for games~\cite{Parikh1983,Parikh1985,PaulyParikh2003,BenthemGhoshLiu2008,BellierBenerecettiMonicaMogavero2023}, for systems with alternating angelic/demonic nondeterminism~\cite{BackWright1998,CavalcantiWoodcockDunne2006,MartinCurtisRewitzky2007}, for systems with alternating forms of concurrency~\cite{Peleg1987} or for abstract interpretations of probabilistic programs~\cite{McIverWeber2005,Weber2008,Tsumagari2012}.

This article contributes to a line of work on algebras of multirelations~\cite{Guttmann2014,FurusawaStruth2015a,BerghammerGuttmann2015,FurusawaStruth2016,BerghammerGuttmann2017} and algebraic languages for these~\cite{FurusawaKawaharaStruthTsumagari2017}, with specific operations for multirelations.
A notable example of an operation on multirelations is their Peleg composition~\cite{Peleg1987}: if $R : X \rto \Pow Y$ relates any $a$ in $X$ with a subset $B$ of $Y$ and if $S : Y \rto \Pow Z$ relates each $b \in B$ with a subset $C_b$ of $Z$, then $R \seq S : X \rto \Pow Z$ relates $a$ with the union of all the $C_b$.
A typical operation on the inner or demonic structure is Peleg's parallel composition of multirelations~\cite{Peleg1987}: if $R : X \rto \Pow Y$ and $S : X \rto \Pow Y$ relate any $a$ in $X$ with subsets $B$ and $C$ of $Y$, respectively, then $R \iu S$ relates $a$ with the inner or demonic choice $B \cup C$.
We refer to this inner operation more neutrally as the inner union of $R$ and $S$.

Further inner operations -- an inner intersection, complementation and duality -- have been defined by Rewitzky~\cite{Rewitzky2003,RewitzkyBrink2006}.
An inner up-closure operation -- if $R$ relates $a$ with $B$ and $B \subseteq C$, then $R$ relates $a$ with $C$ -- plays a key role in Parikh's game logic~\cite{Parikh1983}.
In an up-closed multirelation, each set of inner choices from any given element can be weakened to any superset with more inner choices.
Rewitzky has added a dual down-closure operation, which supports strengthening inner choices to sets with fewer inner choices.
She has also defined an inner preorder, akin to the Smyth preorder of domain theory, which relates $R$ to $S$ if the up-closure of $R$ is contained in that of $S$ and thus compares the inner nondeterminism of these multirelations.

Here, we add new results about the inner structure, the study of which was previously mainly targeted at games and up-closed multirelations.
We close multirelations and multirelations up-to preorder equivalence to quantales and Peleg composition, using tools and techniques from universal algebra.
We also introduce a notion of convex closure, as the intersection of up- and down-closure, together with a corresponding preorder and equivalence, and study their properties.
In Proposition~\ref{proposition.iuiiquantale1} we show that homsets of multirelations form commutative quantales with either inner union or inner intersections as monoidal multiplication.
These are isomorphic with respect to the duality induced by inner complementation, which replaces each set $B$ in each pair $(a,B)$ by its boolean complement.
In Proposition~\ref{proposition.iuiiquantale2} we prove that the up-closed and the down-closed elements in each homset form isomorphic subquantales of the double quantale on the entire homset, in which the inner intersection and the inner union collapses to (outer) intersection, respectively, while the convex-closed elements form an inf-lattice.
In Proposition~\ref{proposition.quotient-quantale} we demonstrate that the quotient quantales on each homset with respect to the equivalences generated by the three preorders on multirelations are isomorphic to the quantales on up-, down- and convex-closed multirelations, respectively.
In addition, we show in Section~\ref{subsection.preorder-special} that the inner preorders become partial orders, and even natural orders with a lattice structure, on certain subclasses of multirelations, and that they coincide on deterministic multirelations.

Up-closed and convex-closed multirelations are relevant to game logics and abstract interpretations of probabilistic programs, respectively.
Down-closure is needed for defining convex-closure, characterising deterministic multirelations in the second article in this trilogy and modal operators on multirelations in the third article~\cite{FurusawaGuttmannStruth2023b,FurusawaGuttmannStruth2023c}.

The interactions of the operations studied in this trilogy are quite complex.
We therefore consider them in concrete extensions and enrichments of $\Rel$, but with a view towards future axiomatic approaches, and generally aim at algebraic proofs.

The technical results in this trilogy of articles have benefitted greatly from working with the Isabelle/HOL proof assistant.
In support of them we have developed a substantial library for multirelations~\cite{GuttmannStruth2023}, which extends a previous one~\cite{FurusawaStruth2015b} from single-homset multirelations to $\Rel$ and adds new results about the inner structure and beyond.
While we have used this library to verify or falsify many conjectures related to this article and to increase our confidence in the correctness of our own definitions and proofs, we did not aim at a complete formalisation.
This article is therefore self-contained without the Isabelle libraries, and not about formalised mathematics.


\section{Relations and Multirelations}
\label{section.relation-multirelation}

We start with recalling the basics of binary relations and multirelations.
See~\cite{Peleg1987,Goldblatt1992,FurusawaStruth2015a,FurusawaStruth2016,FurusawaKawaharaStruthTsumagari2017} for details.
Our algebraic language of concrete relations and multirelations is based on enrichments of the category $\Rel$, with sets as objects and binary relations as arrows.
Among such enrichments are regular categories~\cite{Grillet1971} and Dedekind categories~\cite{OlivierSerrato1980}, but our language is more closely related to relation-algebraic approaches~\cite{SchmidtStroehlein1989,FreydScedrov1990,Schmidt2011}, quantales~\cite{Rosenthal1996} and their extensions with multirelational concepts~\cite{FurusawaKawaharaStruthTsumagari2017}.
We therefore start from concrete definitions in $\Rel$, develop algebraic laws for them and then use algebraic reasoning as much as possible.

The relational calculus is rich and well documented.
Multirelations add a further layer of complexity which is much less explored.
This richness sometimes prevents us from listing all properties used in calculations and proofs -- we often refer to ``standard'' relational properties instead.
We provide a dependency list of relational and multirelational concepts with respect to a small basis in Appendix~\ref{section.basis}.

\subsection{Binary relations}
\label{subsection.binary-relations}

We consider binary relations as arrows in the category $\Rel$ and write $X \rto Y$ for the homset $\Rel(X,Y)$.
The composition of arrows $R : X \rto Y$ and $S : Y \rto Z$ is relational composition $R S = \{ (a,b) \mid \exists c .\ R_{a,c} \wedge S_{c,b} \}$; identity arrows are relations $\Id_X = \{ (a,a) \mid a \in X \}$.
We compose arrows of categories in diagrammatic order, against the direction of function composition, but in the direction of relational composition.
We often drop indices, writing $\Id$ for $\Id_X$ and likewise.

Each homset $\Rel(X,Y)$ forms a complete atomic boolean algebra, and relational composition preserves arbitrary sups in both arguments.
We write $\emptyset_{X,Y}$ for the least and $U_{X,Y}$ for the greatest element in $X \rto Y$, $-R$ for the complement of $R$ and $S - R$ for the relative complement $S \cap -R$.

The relation $R : X \rto X$ is a \emph{test} if $R \subseteq \Id$.
Relational composition of tests is intersection.
Tests form a full subalgebra of $\Rel(X,X)$ for any $X$, a complete atomic boolean algebra.

We consider the following additional basic operations on relations:
\begin{itemize}
\item The \emph{converse} of $R : X \rto Y$ is $R^\converse : Y \rto X, R \mapsto \{ (b,a) \mid R_{a,b} \}$.
\item The \emph{domain} of $R : X \rto Y$ is the test $\dom(R) = \{ (a,a) \mid \exists b .\ R_{a,b} \}$ in $X \rto X$.
      It satisfies $\dom(R) = \Id_X \cap R R^\converse = \Id_X \cap R U_{Y,X}$.
\item The \emph{left residual} of $T : X \rto Z$ and $S : Y \rto Z$ is given by $T / S = \bigcup \{ R : X \rto Y \mid R S \subseteq T \}$.
\item The \emph{right residual} $T \backslash S : X \rto Y$ is given by $T \backslash S = (S^\converse / T^\converse)^\converse$ for $T : Z \rto X$ and $S : Z \rto Y$.
\item The \emph{symmetric quotient} is $\syq{T}{S} = (T \backslash S) \cap (T^\converse / S^\converse)$.
\end{itemize}
Tests and domain elements form the same subalgebras.
The residuals are right adjoints of relational composition.

We also need the following special relations:
\begin{itemize}
\item the \emph{membership relation} $\in_Y : Y \rto \Pow Y$,
\item the \emph{subset relation} $\Omega_Y = {\in_Y} \backslash {\in_Y} = \{ (A,B) \mid A \subseteq B \subseteq Y \}$,
\item the \emph{complementation relation} $C = \syq{\in_Y}{-{\in_Y}} = \{ (A,-A) \mid A \subseteq Y \}$.
\end{itemize}

We use the following properties of relations.
Relation $R : X \rto Y$ is
\begin{itemize}
\item \emph{total} if $\dom(R) = \Id_X$, or equivalently $\Id_X \subseteq R R^\converse$,
\item \emph{univalent}, or a \emph{partial function}, if $R^\converse R \subseteq \Id_Y$,
\item \emph{deterministic}, or a \emph{function}, if it is total and univalent.
\end{itemize}
Functions as deterministic relations in $\Rel$ are of course graphs of functions in $\Set$.

Finally, we write $R|_A$ for the restriction of relation $R$ to domain elements in the set $A$, $R(A)$ for the relational image of $A$ under $R$ and $R(a)$ for $R(\{a\})$.

Relations decompose into unions of partial functions.
Each partial function contains one particular choice of codomain element (as a singleton set) for each domain element with a non-empty relational image.
For $R, S : X \rto Y$, we write $S \subseteq_d R$ if $S$ is univalent, $\dom(S) = \dom(R)$ and $S \subseteq R$.

\begin{lemma}
  \label{lemma.sub-d}
  Let $R : X \rto Y$.
  Then $R = \bigcup_{S \subseteq_d R} S$.
\end{lemma}

\subsection{Multirelations}
\label{SS:multirelations}

A \emph{multirelation} is an arrow $X \rto \Pow Y$ in $\Rel$.
We write $\Mult(X,Y)$ for the homset $X \rto \Pow Y$.

\begin{example}
  The $\in$-relation is a multirelation $X \rto \Pow X$.
  Graphs of nondeterministic functions $X \to \Pow Y$ are deterministic multirelations.
  An instance of this is $\syq{\Id}{\in}$, which relates every element to a singleton set containing it; see units $1_X$ below.
\end{example}

Multirelations can be composed in many ways; see~\cite{FurusawaKawaharaStruthTsumagari2017} for a comparison.
The most relevant to us comes from concurrent dynamic logic~\cite{Peleg1987}.

The \emph{Peleg composition}~\cite{Peleg1987} $\seq : (X \rto \Pow Y) \times (Y \rto \Pow Z) \to (X \rto \Pow Z)$ can be defined in terms of the \emph{Peleg lifting} $(-)_\seq : (X \rto \Pow Y) \to (\Pow X \rto \Pow Y)$ of multirelations~\cite{FurusawaKawaharaStruthTsumagari2017}:
\begin{equation*}
  R \seq S = R S_\seq = \left\{ (a,C) \mid \exists B .\ R_{a,B} \wedge \exists f : Y \to \Pow Z .\ f|_B \subseteq S \wedge C = \bigcup f(B) \right\},
\end{equation*}
where $R_\seq = \{ (A,B) \mid \exists f : X \to \Pow Y .\ f|_A \subseteq R \wedge B = \bigcup f(A) \}$.
In turn, the Peleg lifting satisfies $R_\seq = \dom(R)_\seq \bigcup_{S \subseteq_d R} S_\kleisli$ using the \emph{Kleisli lifting} $(-)_\kleisli : (X \rto \Pow Y) \to (\Pow X \rto \Pow Y)$ given by $R_\kleisli = \syq{{\in} R^\converse {\in}}{\in} = \left\{ (A,B) \mid B = \bigcup R(A) \right\}$.
The units of Peleg composition are the multirelations $1_X = \{ (a,\{a\}) \mid a \in X \}$.

Peleg composition preserves arbitrary unions in its first argument, but only the order in its second one: $R \subseteq S \Rightarrow T \seq R \subseteq T \seq S$.
Thus $\emptyset \seq R = \emptyset$, whereas the right zero law generally fails.
It is not associative either; only $(R \seq S) \seq T \subseteq R \seq (S \seq T)$ holds.
Hence multirelations do not form a category under Peleg composition.
The composition becomes associative if the third factor is univalent or union-closed~\cite{FurusawaKawaharaStruthTsumagari2017} (see also Section~\ref{subsection.union-closed}).


\section{Inner Operations}
\label{section.inner-ops}

The complete atomic Boolean algebra of multirelations $X \rto \Pow Y$ forms an \emph{outer} or \emph{angelic} structure with \emph{outer} operations and properties.
In addition, the boolean algebra $\Pow Y$ on the second components of ordered pairs $(a,A)$ forms a dual \emph{inner} or \emph{demonic} set structure for each $a$, with \emph{inner} operations on multirelations.
The parallel composition of concurrent dynamic logic~\cite{Peleg1987} is an inner union operation; its algebraic properties are well studied~\cite{FurusawaStruth2015a,FurusawaStruth2016}.
A dual inner intersection and an inner complementation that induces this duality have been defined by Rewitzky~\cite{Rewitzky2003}.
She refers to the inner operations as \emph{power union}, \emph{power intersection} and \emph{power negation}.
We now investigate the inner structure at greater detail.

Recall that a \emph{quantale} $(Q,\le,\cdot,1)$ is a complete lattice $(Q,\le)$ and a monoid $(Q,\cdot,1)$ such that $\cdot$ preserves all sups in both arguments, and that quantale morphisms preserve all sups and the monoidal structure~\cite{Rosenthal1990}.
A quantale is \emph{commutative} if $\cdot$ is.

\subsection{Definitions of inner operations}

The \emph{inner union}, \emph{inner intersection}, their units and \emph{inner complementation} are defined, for multirelations $R, S : X \rto \Pow Y$, as
\begin{align*}
  R \iu S & = \{ (a,A \cup B) \mid R_{a,A} \wedge S_{a,B} \}, & \iuone & = \{ (a,\emptyset) \mid a \in X \}, \\
  R \ii S & = \{ (a,A \cap B) \mid R_{a,A} \wedge S_{a,B} \}, & \iione & = \{ (a,Y) \mid a \in X \}, \\
  \icpl{R} & = \{ (a,-A) \mid R_{a,A} \}.
\end{align*}
Algebraically, $\icpl{R} = R C$, where $C$ is the complementation relation from Section~\ref{subsection.binary-relations}.
Further, $\iione = 1 \iu \icpl{1}$ and $\iuone = 1 \ii \icpl{1}$.

\begin{remark}
  \label{rem.basis}
  We do not know relation-algebraic definitions of $\iu$ or $\ii$ and need to add at least one of them, for instance $\iu$, in our multirelational language.
  See also Appendix~\ref{section.basis}.
\end{remark}

\subsection{Algebra of inner operations}
\label{subsection.inner-algebra}

The interaction of $\iu$ with $\seq$ and the outer operations is well known~\cite{FurusawaStruth2015a,FurusawaStruth2016}.
Interactions of $\ii$ usually follow by duality with respect to $\icpl{}$.

\begin{proposition}
  \label{proposition.iuiiquantale1}
  $\Mult_\iu(X,Y) = (\Mult(X,Y),\subseteq,\iu,\iuone)$ and $\Mult_\ii(X,Y) = (\Mult(X,Y),\subseteq,\ii,\iione)$ are commutative quantales.
  Inner complementation $\icpl{} : \Mult_\iu(X,Y) \to \Mult_\ii(X,Y)$ is a quantale isomorphism.
  It is involutive and thus its own inverse.
\end{proposition}

\begin{proof}
  The quantale structure of $\Mult_\iu(X,Y)$ has been checked in~\cite{FurusawaStruth2016}; that of $\Mult_\ii(X,Y)$ follows from the isomorphism we establish next.
  First, $\icpl{}$ is clearly involutive and surjective.
  Second, it is injective because $\icpl{R} = \icpl{S}$ implies $\icpl{\icpl{R}} = \icpl{\icpl{S}}$ and thus $R = S$.
  Third, it preserves inner union, inner intersection, units and arbitrary unions:
  \begin{gather*}
    \icpl{(R \iu S)} = \icpl{R} \ii \icpl{S}, \qquad
    \icpl{(R \ii S)} = \icpl{R} \iu \icpl{S}, \qquad
    \icpl{\iuone} = \iione, \qquad
    \icpl{\iione} = \iuone, \\
    \textstyle \icpl{\bigcup \mathcal{R}} = \bigcup \{ \icpl{R} \mid R \in \mathcal{R} \}.
    \qedhere
  \end{gather*}
\end{proof}

We call $\icpl{}$ the \emph{inner isomorphism} or \emph{inner duality}, in contrast to the \emph{outer isomorphism} or \emph{outer duality} given by boolean complementation $-$.
Properties of $\ii$ thus translate from those of $\iu$ via inner duality, and vice versa.

\begin{remark}
  The quantales $\Mult_\iu(X,Y)$ and $\Mult_\ii(X,Y)$, as powerset structures, are boolean, atomic and completely distributive.
  The inner isomorphism preserves the boolean structure, $\icpl{-R} = -\icpl{R}$, as well as arbitrary intersections.
  In particular, $\icpl{\emptyset} = \emptyset$ and $\icpl{U} = U$, and zero laws $R \iu \emptyset = \emptyset$ and $R \ii \emptyset = \emptyset$ follow immediately from union preservation.
\end{remark}

While $R \subseteq R \iu R$, and dually $R \subseteq R \ii R$, inner union and intersection need not be idempotent and thus do not impose a semilattice structure on $\Mult(X,Y)$.
Thus neither $\Mult_\iu(X,Y)$ nor $\Mult_\ii(X,Y)$ forms a frame or locale, and the quantale order $\subseteq$ is not the natural order on $\iu$ or $\ii$.

\begin{example}
  \label{example.not-idempotent}
  For $R = \{ (a,\{a\}), (a,\{b\}) \}$, $R \iu R = R \cup \{ (a,\{a,b\}) \}$ and $R \ii R = R \cup \{ (a,\emptyset) \}$.
\end{example}

\begin{example}
  \label{example.idempotents}
  The greatest elements $U_{X,Y}$ are idempotents of $\iu$ and $\ii$.
  As in any semigroup, this induces subalgebras in which the $U_{X,Y}$ appear as units.
  The fixpoints of $(-) \iu U$ are precisely the up-closed multirelations~\cite{FurusawaStruth2016}, which appear in game logic~\cite{Parikh1983}.
  By inner duality, the fixpoints of $(-) \ii U$ yield down-closed multirelations.
\end{example}

The subalgebras arising from the idempotents $U$ are studied in Section~\ref{section.up-down}.
Partial functions yield additional idempotents of the inner structure.

\begin{lemma}
  \label{lemma.pfun-idem}
  If $R : X \rto \Pow Y$ is univalent, then $R \iu R = R = R \ii R$.
\end{lemma}

\begin{example}
  The converse does not hold: any $R = \{ (a,A), (a,B) \}$ with $A \subset B$ is idempotent with respect to $\iu$ and $\ii$, but not univalent.
\end{example}

\begin{remark}
  The relationship between $\iu$ and $\ii$ with $\subseteq$ differs from that of the outer operations.
  Implications between $R \subseteq S$, $R \iu S = S$, $R \iu S = R$, $R \ii S = S$ and $R \ii S = R$ can be refuted using small multirelations built from $(a,\emptyset)$, $(a,\{a\})$ and $\emptyset$.
  We obtain $(R \ii S) \iu T \subseteq (R \iu T) \ii (S \iu T)$ and $(R \iu S) \ii T \subseteq (R \ii T) \iu (S \ii T)$, but these properties do not imply order-preservation.
  This further confirms that $\subseteq$ is not the natural order with respect to $\iu$ or $\ii$.
\end{remark}

A \emph{dual} operation can be defined on multirelations as $\dual{R} = -\icpl{R} = -R C$~\cite{Parikh1983,Rewitzky2003}.
It combines the inner and outer one.
It follows that $(-)^d$ is $\subseteq$-reversing and satisfies
\begin{gather*}
  \icpl{R} = -\dual{R}, \qquad
  \dual{\dual{R}} = R, \qquad
  \dual{(R \cap S)} = \dual{R} \cup \dual{S}, \qquad
  \dual{(R \cup S)} = \dual{R} \cap \dual{S}, \\
  \dual{(-R)} = -(\dual{R}), \qquad
  \dual{(\icpl{R})} = \icpl{(\dual{R})}.
\end{gather*}

\subsection{Union-closure}
\label{subsection.union-closed}

Inner union and Peleg composition interact as follows~\cite{FurusawaStruth2015a}:
\begin{gather*}
  (R \iu S) \seq T \subseteq (R \seq T) \iu (S \seq T), \qquad
  R \seq (S \iu T) \subseteq (R \seq S) \iu (R \seq T), \\
  T \iu T \subseteq T \Rightarrow (R \iu S) \seq T = (R \seq T) \iu (S \seq T).
\end{gather*}
The distributivity law generalises.
We define
\begin{equation*}
  \underset{i \in I}{\iU} R_i = \left\{ \left( a, \bigcup_{i \in I} A_i \right) \middle\vert \ \forall i \in I .\ (a,A_i) \in R_i \right\}
\end{equation*}
and call a multirelation $R$ \emph{union-closed} (or \emph{additive}~\cite{Rewitzky2003}) if $\iU_{i \in I} R \subseteq R$ for all $I \neq \emptyset$, or equivalently, if $\dom(S) (\syq{{\in} S^\converse}{\in}) \subseteq R$ for all $S \subseteq R$~\cite{FurusawaKawaharaStruthTsumagari2017}.
Then, for union-closed $S$,
\begin{equation*}
  (\underset{i \in I}{\iU} R_i) \seq S = \underset{i \in I}{\iU} (R_i \seq S).
\end{equation*}

\subsection{Inner determinism, inner univalence}
\label{subsection.atoms}

Relation $U_{X,Y}$ is mapped by $(-) 1_Y$ to
\begin{equation*}
  \iuatoms = U_{X,Y} 1_Y = \{ (a,\{b\}) \mid a \in X \wedge b \in Y \},
\end{equation*}
the set of all (multirelational) \emph{atoms} in $\Mult(X,Y)$.
By inner duality,
\begin{equation*}
  \iiatoms = \{ (a,X-\{b\}) \mid a \in X, b \in Y \}
\end{equation*}
is the set of all \emph{co-atoms}.
Of course, $\icpl{\iuatoms} = \iiatoms$ and $\icpl{\iiatoms} = \iuatoms$.
Atoms allow expressing inner analogues to (outer) determinism, univalence and totality.

Multirelation $R : X \rto Y$ is
\begin{itemize}
\item \emph{inner-univalent} if $R \subseteq \iuatoms \cup \iuone$, that is, $B$ is either a singleton or empty for each $(a,B) \in R$,
\item \emph{inner-total} if $R \subseteq -\iuone$, that is, $B$ is non-empty for each $(a,B) \in R$,
\item \emph{inner deterministic} if it is inner univalent and inner total, that is, $B \subseteq Y$ is a singleton set whenever $R_{a,B}$ for some $a \in X$.
\end{itemize}
Inner deterministic multirelations are obviously subsets of $\iuatoms$.

Inner univalent multirelations thus admit only outer or angelic choices, but not inner ones; they are therefore completely angelic.
Outer univalent multirelations, by contrast, admit only inner or demonic choices, but not outer ones; they are therefore completely demonic.
Inner deterministic multirelations can then be seen as strictly angelic, as all inner choices must be non-empty, and outer deterministic multirelations as strictly demonic, as empty outer choices are impossible.
Inner total multirelations have been called \emph{total}, outer total multirelations \emph{proper}, inner univalent multirelations \emph{angelic} and outer univalent multirelations \emph{demonic} in~\cite{Rewitzky2003,RewitzkyBrink2006}.

The inner univalent, total and deterministic multirelations satisfy fixpoint properties.

\begin{lemma}
  \label{lemma.det-fix-ref}
  ~
  \begin{enumerate}
  \item The inner univalent multirelations are the fixpoints of $(-) \cap (\iuatoms \cup \iuone)$.
  \item The inner total multirelations are the fixpoints of $(-) - \iuone$.
  \item The inner deterministic multirelations are the fixpoints of $(-) \cap \iuatoms$ and $(-) 1^\converse 1$.
  \end{enumerate}
\end{lemma}

\begin{proof}
  We only prove that $R 1^\converse 1 = R$ if and only if $R \cap \iuatoms = R$ for any multirelation $R$.
  This follows immediately from the relational law $P Q \cap S = (P \cap S Q^\converse) Q$ for outer univalent $Q$~\cite{SchmidtStroehlein1989}, instantiated with $Q = 1$: $\iuatoms \cap R = U 1 \cap R = (U \cap R 1^\converse) 1 = R 1^\converse 1$.
\end{proof}

\begin{lemma}
  \label{lemma.inner-det-props}
  Let $R$, $S$, $T$ be multirelations of appropriate types and $R$ inner deterministic.
  Then
  \begin{enumerate}
  \item $R \seq S = R 1^\converse S$,
  \item $R \seq (S \seq T)= (R \seq S) \seq T$,
  \end{enumerate}
\end{lemma}

\begin{proof}
  For (1), $R \seq S = R 1^\converse 1 S_\seq = R 1^\converse (1 \seq S) = R 1^\converse S$, using Lemma~\ref{lemma.det-fix-ref} in the first step.
  For (2), $R \seq (S \seq T) = R 1^\converse (S \seq T) = R 1^\converse S T_\seq = (R 1^\converse S) \seq T = (R \seq S) \seq T$ using (1).
\end{proof}

We mention the following properties without proof (a formal verification can be found in our Isabelle theories).

\begin{lemma}
  \label{lemma.ii-iu-preservation}
  Inner unions preserve outer univalence, inner and outer totality, and outer determinism; inner intersections preserve inner and outer univalence, outer totality and outer determinism.
\end{lemma}


\section{Inner Closures}
\label{section.up-down}

We have mentioned in Example~\ref{example.idempotents} that the fixpoints of $(-) \iu U$ are the up-closed multirelations~\cite{FurusawaStruth2015a}, which play an important role in the semantics of game logics.
The inner isomorphism yields of course a dual notion of down-closure.
We define these notions, add a notion of convex-closure, which appears in the abstract interpretation of probabilistic programs, and discuss the subalgebras induced.

\subsection{Definition of inner closures}

The \emph{(inner) up-closure}, \emph{down-closure} and \emph{convex-closure} of $R : X \rto \Pow Y$ are defined as
\begin{align*}
  \up{R} = R \iu U, \qquad
  \down{R} = R \ii U, \qquad
  \convex{R} = \up{R} \cap \down{R}.
\end{align*}
It is straightforward to check that $\up{(-)}$, $\down{(-)}$ and $\convex{(-)}$ are indeed closure operators.
The subsets of up-, down- and convex-closed multirelations in $\Mult(X,Y)$ are thus
\begin{equation*}
  \Mult_\uparrow(X,Y) = \{ R \mid R \iu U = R \}, \qquad
  \Mult_\downarrow(X,Y) = \{ R \mid R \ii U = R \}, \qquad
  \Mult_\updownarrow(X,Y) = \{ R \mid R = \convex{R} \}.
\end{equation*}

Alternatively, we can use the subset relation $\Omega$, introduced in Section~\ref{subsection.binary-relations}, to define $\up{R} = R \Omega$ and $\down{R} = R \Omega^\converse$.
Expanding definitions shows that
\begin{gather*}
  \up{R} = \{ (a,A) \mid \exists (a,B) \in R .\ B \subseteq A \}, \qquad
  \down{R} = \{ (a,A) \mid \exists (a,B) \in R .\ A \subseteq B \}, \\
  \convex{R} = \{ (a,A) \mid \exists (a,B), (a,C) \in R .\ B \subseteq A \subseteq C \}.
\end{gather*}

As already mentioned, inner-closed multirelations offer greater flexibility with inner choices.
Up-closed multirelations allow weakening inner choices in that one can always add options to any given set of inner choices.
Likewise, with down-closed multirelations one can always strengthen inner choices by disregarding options in any given set.
Convex-closed multirelations therefore enable any range of inner choices bounded by any two inner sets in the multirelation.

Further, we obtain the following duality.

\begin{lemma}
  \label{lemma.updownprops}
  Let $R : X \rto \Pow Y$.
  Then $\icpl{(\up{R})} = \down{(\icpl{R})}$, $\icpl{(\down{R})} = \up{(\icpl{R})}$ and $\icpl{(\convex{R})} = \convex{(\icpl{R})}$.
\end{lemma}

\begin{remark}
  The relationship $(a,B) \in 1_X \iu U_{X,\Pow X}$ if and only if $a \in B$ confirms that ${\in} = \up{1}$ can be defined in the multirelational language.
  See Appendix~\ref{section.basis} for context.
\end{remark}

\subsection{Structure of inner-closed sets}

The inner-closed multirelations form quantales similar to those in Proposition~\ref{proposition.iuiiquantale1}, but part of the inner structure collapses: $\iu$ becomes $\cap$ when multirelations are up-closed~\cite{FurusawaStruth2016}; dually, therefore, $\ii$ becomes $\cap$ when they are down-closed.
First we note the following fact without proof.

\begin{lemma}
  \label{lemma.iuiiquantale}
  Up- and down-closure of multirelations preserve arbitrary unions:
  \begin{equation*}
    \up{\left( \bigcup R \right)} = \bigcup_{S \in R} \up{S} \qquad
    \text{ and } \qquad
    \down{\left( \bigcup R \right)} = \bigcup_{S \in R} \down{S}.
 \end{equation*}
\end{lemma}

These operations need not preserve intersections, but arbitrary intersections of closed elements of any closure operator are of course closed.

Next we present a refinement of Proposition~\ref{proposition.iuiiquantale1}.

\begin{proposition}
  \label{proposition.iuiiquantale2}
  ~
  \begin{enumerate}
  \item $(\Mult_\downarrow(X,Y),\subseteq,\iu,\iuone)$ is a commutative subquantale of $\Mult_\iu(X,Y)$ in which $\ii = \cap$, $\down{\iione} = U$ and $\iuone = \down{\iuone}$.
  \item $(\Mult_\uparrow(X,Y),\subseteq,\ii,\iione)$ is a commutative subquantale of $\Mult_\ii(X,Y)$ in which $\iu = \cap$, $\up{\iuone} = U$ and $\iione = \up{\iione}$.
  \item The maps $\down{(-)} : \Mult(X,Y) \to \Mult_\downarrow(X,Y)$ and $\up{(-)} : \Mult(X,Y) \to \Mult_\uparrow(X,Y)$ are quantale homomorphisms, $\icpl{} : \Mult_\downarrow(X,Y) \to \Mult_\uparrow(X,Y)$ is a quantale isomorphism.
  \item $(\Mult_\updownarrow(X,Y),\bigcap)$ is an inf-lattice.
  \item The map $\convex{(-)} : \Mult(X,Y) \to \Mult_\updownarrow(X,Y)$ is an inf-lattice morphism and $\icpl{} : \Mult_\updownarrow(X,Y) \to \Mult_\updownarrow(X,Y)$ is an inf-lattice automorphism.
  \end{enumerate}
\end{proposition}

\begin{proof}
  For (1)--(3) note that the maps $\down{(-)}$ and $\up{(-)}$ are nuclei: closure operators satisfying $\down{R} \iu \down{S} \subseteq \down{(R \iu S)}$ and $\up{R} \ii \up{S} \subseteq \up{(R \ii S)}$, and in fact
  \begin{equation*}
    \down{R} \iu \down{S} = \down{(R \iu S)} \qquad
    \text{ and } \qquad
    \up{R} \ii \up{S} = \up{(R \ii S)}.
  \end{equation*}
  Hence $\Mult_\downarrow(X,Y)$ is a quantale with composition $\iu$ and $\down{(-)} : \Mult(X,Y) \to \Mult_\downarrow(X,Y)$ is a quantale morphism.
  Likewise $\Mult_\uparrow(X,Y)$ is a quantale with composition $\ii$ and $\up{(-)} : \Mult(X,Y) \to \Mult_\uparrow(X,Y)$ is a quantale morphism~\cite[Theorem 3.3.1]{Rosenthal1990}.
  Moreover, $\down{\iuone} = \iuone$ and $\up{\iione} = \iione$ show unit preservation.
  The map $\icpl{}$ is a quantale isomorphism by Proposition~\ref{proposition.iuiiquantale1} and Lemma~\ref{lemma.updownprops}.
  Further,
  \begin{equation*}
    \up{(R \iu S)} = \up{R} \iu \up{S} = \up{R} \cap \up{S} \qquad
    \text{ and } \qquad
    \down{(R \ii S)} = \down{R} \ii \down{S} = \down{R} \cap \down{S}.
  \end{equation*}
  For $\iu$, this fact is known~\cite{FurusawaStruth2016}.
  That for $\ii$ then follows from inner duality.
  Idempotency of $\iu$ for up-closed multirelations and of $\ii$ for down-closed multirelations and coincidence with $\cap$ are trivial consequences of these facts.

  For (4) and (5), let $S \subseteq \Mult(X,Y)$.
  Then
  \begin{equation*}
      \convex{\left( \bigcap S \right)}
    = \up{\left( \bigcap S \right)} \cap \down{\left( \bigcap S \right)}
    = \left( \bigcap_{R \in S} \up{R} \right) \cap \left( \bigcap_{R \in S} \down{R} \right)
    = \bigcap_{R \in S} \up{R} \cap \down{R}
    = \bigcap_{R \in S} \convex{R}
  \end{equation*}
  using (3).
  This shows that $\convex{(-)}$ preserves arbitrary intersections.
  Moreover if $S \subseteq \Mult_\updownarrow(X,Y)$, then
  \begin{equation*}
    \convex{\left( \bigcap S \right)} = \bigcap_{R \in S} \convex{R} = \bigcap S.
  \end{equation*}
  Hence $\Mult_\updownarrow(X,Y)$ is closed under arbitrary intersections and forms an inf-lattice.
  The automorphism claim about $\icpl{}$ follows from Lemma~\ref{lemma.updownprops} since $\icpl{}$ preserves arbitrary intersections.
\end{proof}

Parts (4) and (5) do not extend to a quantale structure for convex-closed multirelations as there is no operation corresponding to $\iu$ and $\ii$.

The complete sublattices of $\Mult(X,Y)$ need not be boolean: $\Mult_\uparrow(X,Y)$, $\Mult_\downarrow(X,Y)$ and $\Mult_\updownarrow(X,Y)$ are not closed under complementation.

\begin{example}
  \label{example.not-idempotent2}
  The inner intersection of down-closed multirelations, as set-intersection, is idempotent.
  Yet the inner union of down-closed multirelations need not be idempotent: for the multirelation $R$ in Example~\ref{example.not-idempotent}, $\down{R} = \{ (a,\emptyset), (a,\{a\}), (a,\{b\}) \} \subset \down{R} \cup \{ (a,\{a,b\}) \} = \down{R} \iu \down{R}$.
  Dually, while the inner union of up-closed multirelations is idempotent, the inner intersection of up-closed multirelations need not be idempotent: assuming that $R$ is a multirelation on $\{a,b\}$, $\up{R}= R \iu R \subset (R \iu R) \cup \{ (a,\emptyset) \} = \up{R} \ii \up{R}$.
  This shows that set inclusion is still not the natural order on $\Mult_\uparrow(X,Y)$ and $\Mult_\downarrow(X,Y)$.
  Note, however, that $\up{R} \cap \up{S} \subseteq \up{R} \ii \up{S}$ and $\down{R} \cap \down{S} \subseteq \down{R} \iu \down{S}$.
\end{example}

Finally, $\convex{R} \cap \convex{S} = \up{(R \iu S)} \cap \down{(R \ii S)} = \convex{(R \iu S)} \cap \convex{(R \ii S)}$, $\convex{\iuone} = \iuone$, $\convex{\iione} = \iione$, $\convex{1} = 1$, $\convex{U} = U$, $\convex{\emptyset} = \emptyset$ and every univalent multirelation is convex-closed.

\subsection{Inner closures and Peleg composition}

The inner operations, in particular up-closure, have so far been studied primarily in combination with Parikh's composition of multirelations in game logics~\cite{Parikh1983,PaulyParikh2003}.
Note that multirelations under Peleg composition and the outer operations do not form quantales -- or quantaloids, their categorifications -- because Peleg composition is not associative and does not preserve all sups in its second argument.
For similar reasons, and the failure of idempotency of inner union and intersection, they do not form quantales on the inner structure.
See~\cite{FurusawaStruth2015a,FurusawaStruth2016} for more details on these structures.
Here, instead, we relate the inner operations with Peleg composition, which leads to an alternative characterisation of down-closure for multirelations.

\begin{lemma}
  \label{lemma.one-down-lift}
  $(\down{1})_\seq = (1 \cup \iuone)_\seq = \Omega^\converse$, and thus $\down{R} = R \seq \down{1}$ for all $R : X \rto \Pow Y$.
\end{lemma}

\begin{proof}
  First, $(\down{1})_\seq = (1 \cup \iuone)_\seq$ is trivial and
  \begin{equation*}
    (1 \cup \iuone)_\seq = \{ (A,B) \mid \exists f .\ (\forall a \in A .\ f(a) \in \{ \emptyset, \{a\} \}) \wedge B = \bigcup f(A) \} = \{ (A,B) \mid B \subseteq A \} = \Omega^\converse.
  \end{equation*}
  Thus $\down{R} = R \Omega^\converse = R (\down{1})_\seq = R \seq \down{1}$.
\end{proof}

It follows that $\down{1} \seq \down{1} = \down{1}$ and $\Mult_\downarrow(X,Y) = \{ R \mid X \rto \Pow Y \mid R \seq \down{1} = R \}$.

\begin{lemma}
  \label{lemma.eps-comp-up}
  Let $R : X \rto \Pow Y$.
  Then
  \begin{enumerate}
  \item ${\in} \seq \up{R} = \up{R}$,
  \item $\up{R} = R \seq {\in}$ if $R$ is inner total.
  \end{enumerate}
\end{lemma}

\begin{proof}
  For (1), clearly $\up{R} = 1 \seq \up{R} \subseteq {\in} \seq \up{R}$.
  We obtain the converse inclusion by
  \begin{equation*}
    {\in} \seq \up{R} = {\in} \dom(R)_\seq \bigcup_{Q \subseteq_d R} Q_\kleisli = \bigcup_{Q \subseteq_d R} {\in} \dom(Q)_\seq Q_\kleisli \subseteq \up{R}
  \end{equation*}
  if we can show ${\in} \dom(Q)_\seq Q_\kleisli \subseteq \up{Q}$ for univalent $Q$.
  By
  \begin{equation*}
    {\in} \dom(Q)_\seq = {\in} ({\in} \backslash \dom(Q) {\in} \cap \Id) \subseteq {\in} ({\in} \backslash \dom(Q) {\in}) \subseteq \dom(Q) {\in}
  \end{equation*}
  it remains to show $\dom(Q) {\in} Q_\kleisli \subseteq \up{Q}$.
  Since $Q_\kleisli$ is a function and $Q$ is univalent, this is equivalent to
  \begin{align*}
      \dom(Q) {\in}
    & \subseteq Q \Omega (Q_\kleisli)^\converse \\
    & = Q ({\in} \backslash {\in} (Q_\kleisli)^\converse) \\
    & = Q ({\in} \backslash {\in} (\syq{\in}{{\in} Q^\converse {\in}})) \\
    & = Q ({\in} \backslash {\in} Q^\converse {\in}) \\
    & = Q U \cap ({\in} Q^\converse \backslash {\in} Q^\converse {\in})
  \end{align*}
  which follows from $\dom(Q) {\in} \subseteq \dom(Q) U = Q U$ and ${\in} Q^\converse \dom(Q) {\in} \subseteq {\in} Q^\converse {\in}$.

  For (2), suppose $A \neq \emptyset$.
  Then
  \begin{equation*}
    ({\in}_\seq)_{A,B} \Leftrightarrow (\exists f .\ (\forall a \in A .\ a \in f(a)) \wedge B = \bigcup f(A)) \Leftrightarrow (A \subseteq B) \Leftrightarrow \Omega_{A,B}
  \end{equation*}
  where $A \subseteq f(A)$ gives $\Rightarrow$ and $f(a) = B$ gives $\Leftarrow$ of the second equivalence.
  Hence ${\in}_\seq - \iuone^\converse = \Omega - \iuone^\converse$.
  Thus
  \begin{equation*}
    \up{R} = R \Omega = (R - \iuone) \Omega = R (\Omega - \iuone^\converse) = R ({\in}_\seq - \iuone^\converse) = (R - \iuone) {\in}_\seq = R {\in}_\seq = R \seq {\in}
  \end{equation*}
  using that $R$ is inner total.
\end{proof}

\begin{lemma}
  \label{lemma.inner-deterministic-upclosed}
  Let $R$ and $S$ be composable multirelations.
  Then
  \begin{enumerate}
  \item $\down{(R \seq S)} = R \seq \down{S}$ and hence $\down{(\down{R} \seq \down{S})} = \down{R} \seq \down{S}$,
  \item $\up{(R \seq S)} = R \seq \up{S} = \up{R} \seq \up{S}$ if $R$ is inner deterministic.
  \end{enumerate}
\end{lemma}

\begin{proof}
  For (1), $\down{(R \seq S)} = (R \seq S) \seq \down{1} = R \seq (S \seq \down{1}) = R \seq \down{S}$ as $\down{1}$ is union-closed.
  For (2), $\up{(R \seq S)} = R 1^\converse S \Omega = R \seq \up{S} = R \seq ({\in} \seq \up{S}) = (R \seq {\in}) \seq \up{S} = \up{R} \seq \up{S}$ by Lemmas~\ref{lemma.inner-det-props} and~\ref{lemma.eps-comp-up}.
\end{proof}

The Peleg composition of down-closed multirelations is therefore down-closed.

\begin{example}
  Peleg compositions of up-closed multirelations need not be up-closed:
  \begin{equation*}
    \up{\iuone} \seq \up{\iione} = U \seq \iione = \iuone \cup \iione \subset U = \up{\iuone} = \up{(\iuone \seq \iione)}.
  \end{equation*}
  Note that $\iuone$ and $\iione$ are both deterministic.
\end{example}

However, the up-closure of the Peleg composition of up-closed multirelations equals their Parikh composition~\cite{FurusawaStruth2016} (and the co-composition of up-closed multirelations is up-closed, see Section~\ref{section.co-comp}).
As up-closed multirelations are union-closed, their Peleg composition is associative~\cite{FurusawaKawaharaStruthTsumagari2017}.

\begin{example}
  \label{example.up-dist-comp-det}
  The property $\up{(R \seq S)} = \up{R} \seq \up{S}$ from Lemma~\ref{lemma.inner-deterministic-upclosed} does not translate to down-closure: $\down{(1 \seq \emptyset)} = \down{\emptyset} = \emptyset \subset \iuone = \iuone \seq \emptyset = (1 \cup \iuone) \seq \emptyset = \down{1} \seq \down{\emptyset}$.
  Note that $1$ and $\emptyset$ are inner deterministic.
\end{example}

Finally, $\down{R} \seq S = R \seq (\iuone \cup 1) \seq S = R \seq ((\iuone \cup 1) S_\seq) = R \seq (\iuone \cup S)$.
The first step uses a special associativity property in the presence of $\iuone \cup 1$ (proved using Isabelle).


\section{Inner Preorders}
\label{section.inner-preorders}

Example~\ref{example.not-idempotent} shows that $\subseteq$ is not the natural order for $\iu$ and $\ii$.
Proposition~\ref{proposition.iuiiquantale2} shows that restrictions to up- or down-closed relations collapse part of the inner structure.
It is standard to define preorders, equivalences and partial orders based on the inclusion of closed sets.
Here, these preorders compare the inner nondeterminism of multirelations in different ways, while set inclusion obviously compares their outer nondeterminism.
Apart from the obvious interest in such comparisons, this raises the question whether these orders are natural for inner union and inner intersection.
The general answer is negative.

\subsection{Definition of inner preorders}

For $R, S : X \rto \Pow Y$, we define the \emph{Smyth preorder} $\subs$~\cite{Rewitzky2003}, its dual \emph{Hoare preorder} $\subh$ and the \emph{Egli-Milner preorder} $\subem$ as
\begin{equation*}
  R \subs S \Leftrightarrow S \subseteq \up{R}, \qquad
  R \subh S \Leftrightarrow R \subseteq \down{S}, \qquad
  R \subem S \Leftrightarrow R \subh S \wedge R \subs S.
\end{equation*}
Equivalently, $R \subs S \Leftrightarrow \up{S} \subseteq \up{R}$ and dually $R \subh S \Leftrightarrow \down{R} \subseteq \down{S}$.
However,
\begin{equation*}
  \convex{R} \subseteq \convex{S} \Leftrightarrow R \subseteq \convex{S} \Leftrightarrow R \subh S \subs R.
\end{equation*}
Expanding definitions,
\begin{align*}
  R \subs S & \Leftrightarrow (\forall a, C .\ S_{a,C} \Rightarrow \exists B .\ R_{a,B} \wedge B \subseteq C),\\
  R \subh S & \Leftrightarrow (\forall a, B .\ R_{a,B} \Rightarrow \exists C .\ S_{a,C} \wedge B \subseteq C).
\end{align*}

Intuitively, therefore, $R\subs S$ if for every outer choice of a set from a given element with $S$ there is a less nondeterministic outer choice from that element with $R$.
Moreover, $R\subh S$ if for every outer choice of a set from a given element with $R$ there is a more nondeterministic outer choice from that element with $R$.

The following fact is standard.

\begin{lemma}
  \label{lemma.order-emb}
  The map $\down{(-)}$ order-embeds $\subh$ into $\subseteq$ and $\up{(-)}$ order-embeds $\subs$ into $\supseteq$.
\end{lemma}

Moreover, $R \subh S \Leftrightarrow \down{(R \ii S)} = \down{R} \cap \down{S} = \down{R}$ and $R \subs S \Leftrightarrow \up{(R \iu S)} = \up{R} \cap \up{S} = \up{S}$.

\begin{example}
  While $R = R \ii S$ thus implies $R \subh S$ and $S = R \iu S$ implies $R \subs S$, the converse implications, which would be typical for natural orders, do not hold: for $X = \{a\}$, $R = \{ (a,\emptyset) \}$, $S = \{ (a,\{a\}) \}$ and $T = R \cup S$ satisfy $S \subh T$ and $T \subs R$, but $S \ii T = T \neq S$ and $T \iu R = T \neq R$.
\end{example}

We associate equivalences $\eqh$, $\eqs$ and $\eqem$ with $\subh$, $\subs$ and $\subem$ in the standard way by intersecting the preorders with their converses.
Thus
\begin{align*}
  R \eqh S \Leftrightarrow \down{R} = \down{S}, \qquad
  R \eqs S \Leftrightarrow \up{R} = \up{S}, \qquad
  R \eqem S \Leftrightarrow R \subem S \wedge S \subem R.
\end{align*}
It follows that $R \eqem S \Leftrightarrow \convex{R} = \convex{S}$ and therefore $R \eqem \convex{R}$.

\subsection{Algebras of preordered multirelations}

The following results describe the structure of the preorders and the resulting quotient quantales.

\begin{proposition}
  \label{proposition.dmonpreordered}
  ~
  \begin{enumerate}
  \item $(\Mult(X,Y),\subh,\iu,\iuone,\ii,\iione)$ is a preordered commutative double monoid with least element $\emptyset$ and greatest element $U$.
  \item $(\Mult(X,Y),\subs,\iu,\iuone,\ii,\iione)$ is a preordered commutative double monoid with least element $U$ and greatest element $\emptyset$.
  \item $(\Mult(X,Y),\subem,\iu,\iuone,\ii,\iione)$ is a preordered commutative double monoid.
  \item $\icpl{}$ is an order-reversing preordered double monoid isomorphism:
        \begin{equation*}
          R \subh S \Leftrightarrow \icpl{S} \subs \icpl{R}, \qquad
          R \subs S \Leftrightarrow \icpl{S} \subh \icpl{R}, \qquad
          R \subem S \Leftrightarrow \icpl{S} \subem \icpl{R}.
        \end{equation*}
  \item The preorders $\subh$, $\subs$ and $\subem$ are also precongruences with respect to $\cup$, $\up{}$, $\down{}$ and $\convex{}$.
  \item Peleg composition preserves $\subh$, $\subs$ and $\subem$ in its second argument.
  \end{enumerate}
\end{proposition}

\begin{remark}
  Similarly, the three equivalences $\eqh$, $\eqs$ and $\eqem$ are congruences with respect to $\cup$, $\iu$, $\ii$, $\up{}$, $\down{}$, $\convex{}$ and $\icpl{}$, Peleg composition preserves them in its second argument, and the inner isomorphism $\icpl{}$ satisfies $R \eqh S \Leftrightarrow \icpl{R} \eqs \icpl{S}$, $R \eqs S \Leftrightarrow \icpl{R} \eqh \icpl{S}$ and $R \eqem S \Leftrightarrow \icpl{R} \eqem \icpl{S}$.
  Unlike $\subh$ and $\subs$, $\subem$ has no least or greatest element.
\end{remark}

\begin{proposition}
  \label{proposition.quotient-quantale}
  ~
  \begin{enumerate}
  \item $(\Mult(X,Y)/{\eqh},\le_\Ho,\iu_\Ho,1_{\iu_\Ho})$, with $[R] \le_\Ho [S] \Leftrightarrow \down{R} \subseteq \down{S}$, $[R] \iu_\Ho [S] = [R \iu S]$, $1_{\iu_\Ho} = \{\iuone\}$, $[R] \ii_\Ho [S] = [\down{R} \cap \down{S}]$ and $1_{\ii_\Ho} = \{U\}$, is isomorphic to $\Mult_\downarrow(X,Y)$.
  \item $(\Mult(X,Y)/{\eqs},\le_\Sm,\ii_\Sm,1_{\ii_\Sm})$, with $[R] \le_\Sm [S] \Leftrightarrow \up{R} \supseteq \up{S}$, $[R] \iu_\Sm [S] = [\up{R} \cap \up{S}]$, $1_{\iu_\Sm} = \{U\}$, $[R] \ii_\Sm [S] = [R \ii S]$ and $1_{\ii_\Sm} = \{\iione\}$, is isomorphic to $\Mult_\uparrow(X,Y)$.
  \item $(\Mult(X,Y)/{\eqem},\le_\EM,\iu_\EM,1_{\iu_\EM})$, with $[R] \le_\EM [S] \Leftrightarrow \convex{R} \subseteq \convex{S}$, $[R] \iu_\EM [S] = [R \iu S]$, $1_{\iu_\EM} = \{\iuone\}$, $[R] \ii_\EM [S]=[R \ii S]$ and $1_{\ii_\EM} = \{\iione\}$, is isomorphic to $\Mult_\updownarrow(X,Y)$.
  \item $\icpl{[R]_\Ho} = [\icpl{R}]_\Sm$, $\icpl{[R]_\Sm} = [\icpl{R}]_\Ho$ and $\icpl{[R]_\EM} = [\icpl{R}]_\EM$.
  \end{enumerate}
\end{proposition}

\begin{proof}
  The following diagram illustrates the construction.
  \begin{equation*}
    \begin{tikzcd}
      \Mult(X,Y)/{\eqs} \ar[d, "\iota'"'] & \Mult(X,Y) \ar[r, "\varphi"] \ar[l, "\varphi'"'] \ar[dr, "\down{(-)}"'] \ar[dl, "\up{(-)}"] & \Mult(X,Y)/{\eqh} \ar[d, "\iota"] \\[3ex]
      \Mult_\uparrow(X,Y) \ar[rr, "\icpl{}"] & & \Mult_\downarrow(X,Y)
    \end{tikzcd}
  \end{equation*}
  By Proposition~\ref{proposition.iuiiquantale2}, $\down{(-)} : \Mult(X,Y) \to \Mult_\downarrow(X,Y)$ is a quantale homomorphism.
  It follows from standard results of universal algebra~\cite[Theorem 6.7]{BurrisSankappanavar1981} that its kernel, $\eqh$, is a congruence that preserves the quantale operations.
  The associated quotient algebra $\Mult(X,Y)/{\eqh}$ is an algebra with the same signature and quantale operations defined as in (1).
  The natural map $\varphi : \Mult(X,Y) \to \Mult(X,Y)/{\eqh}$, which associates each element with its equivalence class, is thus a bijective quantale morphism~\cite[Theorem 6.10]{BurrisSankappanavar1981}.
  By~\cite[Theorem 6.12]{BurrisSankappanavar1981}, there is then an isomorphism $\iota : \Mult(X,Y)/{\eqh} \to \Mult_\downarrow(X,Y)$, here given by $\iota : [R] \mapsto \down{R}$, such that the above diagram commutes.

  The order isomorphism between $\le_\Ho$ and $\subseteq$ is established by the fact that $[R] \le_\Ho [S] \Leftrightarrow R \subh S$ (by definition of $\eqh$) and by Lemma~\ref{lemma.order-emb}.
  It remains to consider inner intersection $\ii_\Ho$ and its unit.
  The inner intersection $[R] \ii_\Ho [S] = [R \ii S]$ is mapped by the isomorphism to $\down{(R \ii S)} = \down{R} \ii \down{S} = \down{R} \cap \down{S}$.
  Finally, the associated unit is mapped to $U$.

  The proofs for (2) and (3) are similar.
  The proof of (3) uses the kernel $\eqem$ of the inf-lattice morphism $\convex{}$.
  Finally, (4) is obvious.
\end{proof}

\begin{remark}
  By construction, $R \iu R \eqh R \cap R = R$ and $R \ii R \eqs R \cap R = R$ due to the collapse of structure.
  Yet $R \ii R \eqh R \cap R = R$ and $R \iu R \eqs R \cap R = R$ need not hold.
  This is a consequence of Examples~\ref{example.not-idempotent} and~\ref{example.not-idempotent2}, recalling that $\down{(R \iu S)} = \down{R} \iu \down{S}$ and dually $\up{(R \ii S)} = \up{R} \ii \up{S}$.
\end{remark}

The question thus remains whether $\subh$ and $\subs$ are natural orders on certain subalgebras of $\Mult(X,Y)$.
We provide an answer in the next section.

We conclude this section with a collection of properties, proved using Isabelle.

\begin{lemma}
  \label{lemma.preorder-inner}
  Let $R, S, T : X \rto \Pow Y$.
  Then
  \begin{enumerate}
  \item $R \subseteq S \Rightarrow R \cap T \subh S \cap T \subs R \cap T$,
  \item $R \ii S \subh R \subh R \iu R$, $R \ii R \subs R \subs R \iu S$ and $R \ii R \subem R \subem R \iu R$,
  \item $R \ii S \subh R \iu S$, $R \ii S \subs R \iu S$ and $R \ii S \subem R \iu S$,
  \item $R \ii S$ is the inf and $R \cup S$ the sup of $R$ and $S$, up-to $\eqh$, with respect to $\subh$,
  \item $R \iu S$ is the sup and $R \cup S$ the inf of $R$ and $S$, up-to $\eqs$, with respect to $\subs$.
  \end{enumerate}
\end{lemma}

Items (4) and (5) may seem to contradict Proposition~\ref{proposition.quotient-quantale}.
Yet preorders equipped with sups and infs up-to preorder-equivalence need not form lattices, and these facts are not related to the failure of idempotence of $\iu$ and $\ii$.
In particular, recall that $R \ii S \eqh R \cap S \eqs R \iu S$ are in fact infs, while the problematic operations $\iu$ and $\ii$ are ignored in $\subh$ and $\subs$, respectively -- $\cup$ is used instead in both preorders.

\subsection{Inner preorders on special multirelations}
\label{subsection.preorder-special}

In this section, we consider $\subh$, $\subs$ and $\subem$ on subclasses of multirelations.
First we consider cases for which these preorders become partial orders.

\begin{proposition}~
  \label{proposition.subh-subs-order}
  \begin{enumerate}
  \item On inner deterministic multirelations, preorders $\subh$ and $\subs$ coincide with $\subseteq$ and $\supseteq$, respectively, whence $\subem$ is the discrete order.
  \item Preorder $\subem$ is a partial order on inner univalent multirelations.
  \item Preorders $\subh$, $\subs$ and $\subem$ are partial orders on outer univalent multirelations.
  \item The three partial orders coincide on outer deterministic multirelations.
  \end{enumerate}
\end{proposition}

\begin{proof}
  Note that
  \begin{equation*}
    1 (\Id \cup -\Omega^\converse) 1^\converse = 1 1^\converse \cup 1 (-\Omega^\converse) 1^\converse = \Id \cup -(1 \Omega^\converse 1^\converse) = \Id \cup -(1 {\in}^\converse) = \Id \cup -\Id = U .
  \end{equation*}
  Hence $\iuatoms^\converse \iuatoms = 1^\converse U 1 \subseteq 1^\converse 1 (\Id \cup -\Omega^\converse) 1^\converse 1 \subseteq \Id \cup -\Omega^\converse$.
  Thus $\Omega^\converse \cap \iuatoms^\converse \iuatoms \subseteq \Id$.

  For (1), assume that $R \subh S$ for inner deterministic $R$ and $S$.
  Then
  \begin{equation*}
    R = R \cap \iuatoms \subseteq \down{S} \cap \iuatoms = (S \cap \iuatoms) \Omega^\converse \cap \iuatoms = S (\Omega^\converse \cap \iuatoms^\converse \cap \iuatoms) = S (\Omega^\converse \cap \iuatoms^\converse \iuatoms) \subseteq S .
  \end{equation*}
  The converse implication follows by $R \subseteq S \subseteq \down{S}$.
  Moreover, from $R \subs S$ we obtain
  \begin{equation*}
    S = S \cap \iuatoms \subseteq \up{R} \cap \iuatoms = (R \cap \iuatoms) \Omega \cap \iuatoms = R (\Omega \cap \iuatoms^\converse \cap \iuatoms) = R (\Omega \cap \iuatoms^\converse \iuatoms) \subseteq R .
  \end{equation*}
  The converse implication follows by $S \subseteq R \subseteq \up{R}$.

  For (2), we prove antisymmetry of $\subem$ in the inner univalent case.
  Suppose $R \subem S$ and $S \subem R$ for inner univalent $R$ and $S$.
  We show $R \subseteq S$.
  The assumption implies that
  \begin{equation*}
    R = R \cap \down{S} \cap \up{S} = R \cap S \Omega^\converse \cap S \Omega \subseteq S (\Omega^\converse \cap S^\converse R) \cap S \Omega.
  \end{equation*}
  Since $R$ and $S$ are inner univalent, we have $S^\converse R \subseteq (\iuone \cup \iuatoms)^\converse (\iuone \cup \iuatoms) \subseteq U \iuone \cup \iuone^\converse \iuatoms \cup \iuatoms^\converse \iuatoms$.
  Hence, by distributivity, it suffices to consider the following three cases:
  \begin{itemize}
  \item $S (\Omega^\converse \cap U \iuone) \cap S \Omega \subseteq U \iuone \cap S \Omega = S (\Omega \cap U \iuone) \subseteq S \Omega \iuone^\converse \iuone = S (\down{\iuone})^\converse \iuone = S \iuone^\converse \iuone \subseteq S$ using $\iuone^\converse \iuone \subseteq \Id$.
  \item $S (\Omega^\converse \cap \iuone^\converse \iuatoms) \cap S \Omega = \emptyset \subseteq S$ using $\Omega^\converse \cap \iuone^\converse \iuatoms \subseteq \iuone^\converse (\iuatoms \cap \iuone \Omega^\converse) \subseteq U (\iuatoms \cap \down{\iuone}) = U (\iuatoms \cap \iuone) = U \emptyset = \emptyset$.
  \item $S (\Omega^\converse \cap \iuatoms^\converse \iuatoms) \cap S \Omega \subseteq S \Id = S$.
  \end{itemize}
  The proof of $S\subseteq R$ follows along similar lines.

  For (3), we first prove antisymmetry of $\subh$.
  Suppose $R \subh S$ and $S \subh R$, that is, $R \eqh S$, for univalent $R$ and $S$.
  Then $S^\converse R \subseteq S^\converse \down{S} = S^\converse S \Omega^\converse \subseteq \Omega^\converse$ and likewise $R^\converse S \subseteq \Omega^\converse$ by univalence of $R$ and $S$.
  Therefore $S^\converse R \subseteq \Omega^\converse \cap \Omega = \Id$.
  Thus $R = R \cap \down{S} = R \cap S \Omega^\converse \subseteq S S^\converse R \subseteq S$ and $S \subseteq R$ follows by opposition.
  This proves $R = S$.

  Antisymmetry of $\subs$ is proved along similar lines.
  Antisymmetry of $\subem$ is then immediate.

  For (4), suppose $R$ and $S$ are outer deterministic.
  Then $\subh$ and $\subs$ coincide because
  \begin{equation*}
    R \subseteq \down{S} \Leftrightarrow S^\converse \subseteq \Omega^\converse R^\converse \Leftrightarrow S \subseteq \up{R}
  \end{equation*}
  and the claim for $\subem$ follows.
\end{proof}

It is immediate from the proof of Proposition~\ref{proposition.subh-subs-order} that, for $R$, $S$ outer univalent or inner deterministic,
\begin{equation*}
  R \eqh S \Leftrightarrow R \eqs S \Leftrightarrow R \eqem S \Leftrightarrow R = S.
\end{equation*}

Next we point out a case when $\subh$ and $\subs$ become natural orders.

\begin{lemma}
  \label{lemma.nat-order}
  Let $R$ and $S$ be outer univalent.
  Then
  \begin{equation*}
    R \subs S \Leftrightarrow R \iu S = S \qquad
    \text{ and } \qquad
    R \subh S \Leftrightarrow R \ii S = R.
  \end{equation*}
\end{lemma}

\begin{proof}
  Assuming $R \subh S$ we have $R \subseteq \down{S}$ and hence $R \subseteq \down{R} \cap \down{S} = \down{(R \ii S)}$ by Proposition~\ref{proposition.iuiiquantale2}.
  Thus $R \subh R \ii S$.
  By Lemma~\ref{lemma.preorder-inner}, $R \ii S \subh R$.
  Since $R \ii S$ is outer univalent by Lemma~\ref{lemma.ii-iu-preservation}, we obtain $R = R \ii S$ by Proposition~\ref{proposition.subh-subs-order}.
  The converse implication is immediate by Lemma~\ref{lemma.preorder-inner}.

  The proof for $\subs$ is similar.
\end{proof}

\begin{proposition}
  The outer deterministic multirelations form a lattice with respect to $\subh$ (which is equal to $\subs$ and $\subem$) with sup $\iu$ and inf $\ii$.
\end{proposition}

\begin{proof}
  Outer deterministic multirelations are closed with respect to $\iu$ and $\ii$ by Lemma~\ref{lemma.ii-iu-preservation}.
  Since $\iu$ and $\ii$ are associative and commutative, it remains to verify the absorption laws.
  First, $R \iu (R \ii S) = R$ is equivalent to $R \ii S \subs R$ by Lemma~\ref{lemma.nat-order}, which is $R \ii S \subh R$ by Proposition~\ref{proposition.subh-subs-order}, which holds by Lemma~\ref{lemma.preorder-inner}.
  Second, $R = R \ii (R \iu S)$ is equivalent to $R \subh R \iu S$ by Lemma~\ref{lemma.nat-order}, which is $R \subs R \iu S$ by Proposition~\ref{proposition.subh-subs-order}, which holds by Lemma~\ref{lemma.preorder-inner}.
\end{proof}

Deterministic multirelations are isomorphic to relations, and the inner preorders allow comparing their nondeterminism.

\begin{example}
  \label{example.subh-subs-order}
  Let $X = \{a,b,c\}$ and $R, S : X \rto \Pow X$ with $R = \{ (a,\{a\}), (a,\{a,b,c\}) \}$ and $S = R \cup \{ (a,\{a,b\}) \}$.
  Then $R \eqh S$ and $R \eqs S$ but $R \neq S$.
  Hence $\subh$ or $\subs$ are not partial orders on inner total multirelations.
  With the same example, $U R \eqh U S$ and $U R \eqs U S$ but $U R \neq U S$ shows that requiring totality does not suffice either.

  Moreover, on a one-element set all multirelations are inner univalent, $1 \eqh U$ and $-1 \eqs U$ but $1 \neq U \neq -1$.
  Hence inner univalence is also not enough to force a partial order.

  This example also shows that $\subem$ is not a partial order on total or inner total multirelations.
\end{example}

\begin{example}
  Since $\emptyset \subh 1$ and $1 \subs \emptyset$ but neither $\emptyset \subs 1$ nor $1 \subh \emptyset$ hold, preorders $\subh$ and $\subs$ are incomparable for univalent, inner univalent, inner total or inner deterministic multirelations.
  Since $1 \subh \iuone \cup 1$ and $\iuone \cup 1 \subs \iuone$ but neither $1 \subs \iuone \cup 1$ nor $\iuone \cup 1 \subh \iuone$, preorders $\subh$ and $\subs$ are incomparable for total multirelations.
\end{example}

\begin{example}
  In the deterministic case, $\subh$ need not coincide with $\subseteq$.
  For instance, $\{ (a,\emptyset) \}\subh \{ (a,\{a\}) \}$, but the two relations are disjoint.
\end{example}

\subsection{Decomposition of multirelations}

As an application of inner preorders, we present a decomposition theorem for multirelations.
We write $S \sqsubseteq_{\downarrow d} R$ if $S$ is univalent and inner deterministic, $\dom(S) = \dom(R - \iuone)$ and $S \subh R$.

\begin{lemma}
  \label{lemma.pfun-dec}
  Let $R : X \rto \Pow Y$ be univalent.
  Then $R = \dom(R) \iU_{S \sqsubseteq_{\downarrow d} R} S$ and each $S \sqsubseteq_{\downarrow d} R$ is isomorphic to a partial function from $X$ to $Y$.
\end{lemma}

This and Lemma~\ref{lemma.sub-d} yields the following decomposition theorem for multirelations.

\begin{proposition}
  \label{proposition.dec-thm}
  Let $R : X \rto \Pow Y$.
  Then $R = \dom(R) \bigcup_{S \subseteq_d R} \iU_{T \sqsubseteq_{\downarrow d} S} T$.
\end{proposition}

\begin{remark}
  Alternatively, we could define $S\sqsubseteq_{\downarrow d} R$ if $S$ is deterministic and inner univalent and $S \subh R$.
  Unlike with $\sqsubseteq_{\downarrow d}$, pairs of the form $(a,\emptyset)$ are now included.
  Both definitions yield a decomposition theorem, but including such pairs in decompositions is unnecessary.
\end{remark}


\section{Co-composition and Intersection-Closure}
\label{section.co-comp}

Recall the interaction of inner union and Peleg composition:
\begin{gather*}
  (R \iu S) \seq T \subseteq (R \seq T) \iu (S \seq T), \qquad
  R \seq (S \iu T) \subseteq (R \seq S) \iu (R\seq T), \\
  T \iu T \subseteq T \Rightarrow (R \iu S) \seq T = (R \seq T) \iu (S \seq T).
\end{gather*}
To obtain similar properties of $\ii$ by inner duality we need to connect $\icpl{}$ and $\seq$.
The relationship
\begin{equation*}
  \textstyle \icpl{(R \seq S)} = \{ (a,C) \mid \exists B .\ R_{a,B} \wedge \exists f .\ f|_B \cap S = \emptyset \wedge C = \bigcap f(B) \}.
\end{equation*}
motivates defining a \emph{co-composition}
\begin{equation*}
  \textstyle R \seqint S = \icpl{(R \seq \icpl{S})} = \{ (a,C) \mid \exists B .\ R_{a,B} \wedge \exists f .\ f|_B \subseteq S \wedge C = \bigcap f(B) \}.
\end{equation*}
It follows immediately that $R \seq S = \icpl{(R \seqint \icpl{S})}$, $\emptyset \seqint R = \emptyset = R \seqint \emptyset$, $1 \seqint R = R$, $\icpl{R} = R \seqint \icpl{1}$ and $\icpl{1} \seqint \icpl{1} = 1$.
But $\seqint$ does not have a right unit because $\iuone \seqint R = \iione$.

We also obtain $R \seqint \iuone = \icpl{(R \seq \iione)}$ and $R \seqint \iione = \icpl{(R \seq \iuone)}$ and it follows that $R \seq \iuone \subseteq R \ii \icpl{R}$ and $R \seqint \iione \subseteq R \iu \icpl{R}$.

The inner isomorphism tells us that the interaction of co-composition with the outer operations is as weak as that of Peleg composition.
Co-composition preserves $\cup$ in its first argument and $\subseteq$ in its second one.
Moreover $R \seqint (S \ii T) \subseteq (R \seqint S) \ii (R \seqint T)$ and $(R \iu S) \seqint T \subseteq (R \seqint T) \ii (S \seqint T)$, and $(R \iu S) \seqint T = (R \seqint T) \ii (S \seqint T)$ whenever $T \ii T \subseteq T$.

Intersection-closure is defined analogously to union-closure with respect to the inner intersection $\iI_{i \in I} R_i$ of a family of multirelations $R_i$.
The isomorphism $\icpl{}$ extends from finite inner union and intersections to arbitrary ones.
For intersection-closed $T$, we have $(\iU_{i \in I} R_i) \seqint T = \iI_{i \in I} (R_i \seqint T)$ for each $I$.

Intersection-closed multirelations have been called ``multiplicative'' in~\cite{Rewitzky2003,RewitzkyBrink2006}, noting distributivity properties of Parikh composition over intersections.
Here we obtain distributivity results of Peleg (co-)composition over inner unions.
The dual additivity property studied by~\cite{Rewitzky2003,RewitzkyBrink2006}, however, differs from union-closure.

Down-closed multirelations are intersection-closed.
Moreover, Lemma~\ref{lemma.one-down-lift} implies that $\up{R} = (\icpl{R}) \seqint \up{(\icpl{1})}$ by inner duality using Lemma~\ref{lemma.updownprops}.
Note that $\up{(\icpl{1})} = \icpl{(\down{1})} = (\iione \cup \icpl{1})$, so that $\up{R} = (\icpl{R}) \seqint (\iione \cup \icpl{1})$.
Thus $\up{(R \seqint S)} = R \seqint \up{S}$ by the inner isomorphism.

The interaction of co-composition with the inner preorders is weak: operation $\seqint$ preserves $\subh$, $\subs$, $\subem$, $\eqh$, $\eqs$ and $\eqem$ in its second argument.
Furthermore, $R \subs R \seqint \iione$.


\section{Conclusion}

We have studied the inner structure of multirelations and their interaction with Peleg composition in the language of relation algebra and universal algebra.
We have considered in particular the operations of inner and outer union, intersection and complementation, a duality between the inner and outer levels, up-closures and down-closures of multirelations and the associated preorders and equivalences, with a view on their structure and future algebraic axiomatisations.

In the second article in this trilogy~\cite{FurusawaGuttmannStruth2023b} we use the results obtained here to study inner and outer univalent and deterministic multirelations and their categories, and introduce determinisation maps from multirelations to inner and outer deterministic multirelations.
In the third article~\cite{FurusawaGuttmannStruth2023c} we use these maps to develop an algebraic approach to modal operators on multirelations, related to previous work by Nerode and Wijesekera~\cite{NerodeWijesekera1990} and Goldblatt~\cite{Goldblatt1992}.

Based on the multirelational language of concrete relations and multirelations and its properties in this work, an axiomatic extension of the relation algebra used in this article with multirelational operations is the most natural continuation.
It also remains to consider other families of multirelations, in particular up-closed and convex-closed ones, and multiplications other than Peleg composition in relationship to the approach in this article, beyond the initial work by Rewitzky~\cite{Rewitzky2003}.
Convex-closed multirelations have so far received little attention, but seem relevant to the semantics and verification of probabilistic programs with probabilistic distribution transformers, at least to abstract interpretations of these~\cite{McIverWeber2005,Weber2008}.

\paragraph{Acknowledgement}
Hitoshi Furusawa and Walter Guttmann thank the Japan Society for the Promotion of Science for supporting part of this research through a JSPS Invitational Fellowship for Research in Japan.


\bibliographystyle{alpha}
\bibliography{multirel}

\appendix

\section{Basis}
\label{section.basis}

Almost every operation in this article can be defined in terms of a basis of 6 operations that mix the relational and the multirelational language: the relational operations $-$, $\cap$, $/$ and the multirelational operations $1$, $\iu$, $\seq$:

\begin{multicols}{3}
\begin{itemize}
\item $R \cup S = -(-R \cap -S)$
\item $R - S = R \cap -S$
\item $\emptyset = R \cap -R$
\item $U = -\emptyset$
\item $\up{R} = R \iu U$
\item ${\in} = \up{1}$
\item $\Id = 1 / 1$
\item $R^\converse = -(-\Id / R)$
\item $S R = -(-S / R^\converse)$
\item $R \backslash S = (S^\converse / R^\converse)^\converse$
\item $\syq{R}{S} = (R \backslash S) \cap (R^\converse / S^\converse)$
\item $R_\kleisli = \syq{{\in} R^\converse {\in}}{\in}$
\item $\Omega = {\in} \backslash {\in}$
\item $C = \syq{\in}{-\in}$
\item $\icpl{R} = R C$
\item $R \ii S = \icpl{(\icpl{R} \iu \icpl{S})}$
\item $\down{R} = X \ii U$
\item $\convex{R} = \up{R} \cap \down{R}$
\item $\iuone = 1 \ii \icpl{1}$
\item $\iione = \icpl{\iuone}$
\item $\dual{R} = -\icpl{R}$
\item $R \seqint S = \icpl{(R \seq \icpl{S})}$
\item $R_\seq = ((\syq{1^\converse{\in}}{\in}) \seq 1^\converse R 1) \Id_\kleisli$
\item $\iuatoms = U 1$
\item $\iiatoms = \icpl{\iuatoms}$
\item $\dom(R) = \Id \cap R R^\converse$
\item $R \subs S \Leftrightarrow S \subseteq \up{R}$
\item $R \subh S \Leftrightarrow R \subseteq \down{S}$
\item $R \subem S \Leftrightarrow R \subh S \wedge R \subs S$
\end{itemize}
\end{multicols}

If $\seq$ is extended to relations, the simpler definition $R_\seq = \Id \seq R$ may be used.
Alternatively, we could of course replace Peleg composition by Peleg lifting in the basis.
Finally, relational $/$ is required to define some of the operations in our list as it is the only operation in the basis that can change types.
We have so far not attempted to axiomatise the basic operations in the sense of (heterogeneous) relation algebra~\cite{SchmidtStroehlein1989}, concurrent dynamic algebra~\cite{FurusawaStruth2016} or likewise.

\end{document}